\documentclass[conference]{IEEEtran}
\IEEEoverridecommandlockouts

\usepackage[letterpaper, top=.75in, bottom=.75in, left=0.75in, right=0.75in]{geometry}

\newgeometry{top=1in, bottom=.75in, left=.75in, right=.75in}

\usepackage{amsmath,amssymb,amsfonts}
\usepackage{algorithm}
\usepackage{algpseudocode}

\usepackage{graphicx}
\usepackage{textcomp}
\usepackage{xcolor}
\def\BibTeX{{\rm B\kern-.05em{\sc i\kern-.025em b}\kern-.08em
    T\kern-.1667em\lower.7ex\hbox{E}\kern-.125emX}}
\usepackage{cite}  

\usepackage{enumerate}
\usepackage{cleveref}
\usepackage{comment}
\usepackage{algorithm}
\usepackage{booktabs}
\usepackage{upgreek}
\usepackage{romannum}

\usepackage{times}
\usepackage{courier}
\usepackage{helvet}
\usepackage[hyphens]{url}
\usepackage{caption}
\usepackage{subcaption}
\usepackage{mathtools}

\usepackage{tikz}
\usetikzlibrary{arrows.meta, positioning, fit, backgrounds, decorations.pathreplacing}

\usepackage{amsthm}
\theoremstyle{plain}
\newtheorem{theorem}{Theorem} 
\newtheorem{lemma}{Lemma}

\newtheorem{corollary}{Corollary}
\newtheorem{property}{Property}

\theoremstyle{definition}
\newtheorem{definition}{Definition}
\newtheorem{example}{Example}
\newtheorem{assumption}{Assumption}[subsection]

\theoremstyle{remark}
\newtheorem{remark}{Remark}[subsection]

\begin{document}
\title{
Port-Transversal Barriers: Graph-Theoretic Safety for Port-Hamiltonian Systems
}

\author{Chi Ho Leung and Philip E. Par\'{e}*
    \thanks{*Chi Ho Leung and Philip E. Par\'e are with the Elmore Family School of Electrical and Computer Engineering, Purdue University, USA.
    E-mail: leung61@purdue.edu, 
            philpare@purdue.edu. 
    This material is based upon work supported in part by the US National Science Foundation (NSF-ECCS \#2238388).
    }
}

\maketitle

\begin{abstract}
Synthesizing control barrier functions (CBFs) for complex nonlinear systems remains difficult, particularly when safety constraints have high relative degree with respect to the input.
In this work, we study the class of composite port-Hamiltonian systems as a network of interconnected subsystems, and this network admits a natural graph structure.
We propose a novel graph-theoretic framework for the CBF synthesis problem.
The framework addresses the problem of CBF construction for complex nonlinear systems by offering a relative-degree constraint induced by graph distance and a barrier reshaping rule informed by the graph topology.
Furthermore, we perform a feasibility analysis of the reshaped barrier under a bounded input constraint and identify the sufficient conditions for feasibility.
Lastly, the proposed framework is validated on an LC ladder network, where the CBF feasibility on a capacitor charge constraint depends on where the input enters the network.
\end{abstract}


\section{Introduction}

The design of safety filters for complex nonlinear systems remains a challenge despite an extensive research effort.
In particular, safety constraints could result in a relative degree greater than one with respect to the input.
In this case, the standard CBF quadratic program (QP) safety filter~\cite{ames2019control,ames2017cbf} is not directly applicable.

Mitigation strategies include high-order CBF, control barrier value function, and energy-aware CBF.
Each of these strategies operates under different principles and comes with its own set of trade-offs in computation complexity and modeling assumptions.
For instance, high-order CBFs require the exact model knowledge of the relative degree $r$-th repeated Lie derivatives and the model error compounds with each derivative~\cite{nguyen2016exponential, xiao2021high}. 
Further, control barrier value functions suffer from the curse of dimensionality when solving for the CBF via the variational inequality inherited from the Hamilton-Jacobi reachability problem~\cite{choi2021robust}. 
Lastly, energy-aware CBF construction is tied to canonical mechanical coordinates, and a bounded input feasibility analysis has not yet been attempted~\cite{singletary2021safety, califano2024effect}.
However, the body of literature that addresses the high relative degree issue by leveraging the interconnection structure of the underlying system remains sparse.
The approach proposed in this work can be seen as a generalization of energy-aware CBFs to the class of interconnected port-Hamiltonian systems.

The port-Hamiltonian theory provides a physically interpretable and network-based approach to model complex nonlinear systems as interconnected subsystems~\cite{duindam2009modeling, neary2023compositional}.
This interconnection structure induces a graph over the subsystems, whose properties we analyze for the CBF construction.
The graph theoretic analysis is composed of two parts: (i) relative degree constraints induced by graph distance and (ii) control barrier reshaping rule informed by the graph's topology.
Lastly, we perform a feasibility analysis of the barrier under a bounded input constraint.


\subsection*{Notation}
The notation \(\mathbb{R}\) denotes the real number line. 
Vectors in \(\mathbb{R}^n\) are column vectors. 
$[A]_{ij}$ denotes the $i,j$ entry of a matrix $A$.
$C^1(\mathcal{D})$ denotes the class of continuously differentiable functions in the domain $\mathcal{D}\subseteq \mathbb{R}^n$.
We denote $(s)^+ = \max(0,s)$.
The gradient of a scalar function $H$ with respect to $x$ is denoted as $\nabla_x H$.
For a set $\mathcal S$, its interior is denoted by $\operatorname{int}(\mathcal S)$ and its boundary is denoted by $\partial \mathcal S$.

\section{Problem Formulation}
\label{subsec:problem}

\noindent\textit{System Dynamics.}
Consider $N$ port-Hamiltonian subsystems:
\begin{equation}
    \begin{split}\label{eq:pH-subsystems}
        \dot{x}_i &= \bigl[J_i(x_i) - R_i(x_i)\bigr]\nabla H_i(x_i) + G_i(x_i)u_i + w_i,\\
        y_i &= G_i^\top \nabla H_i(x_i),\quad z_i = \nabla H_i,
    \end{split}
\end{equation}
with $x_i \in \mathcal{D}_i \subset \mathbb{R}^{n_i}$,
    $u_i \in \mathcal{U}_i \subset \mathbb{R}^{m_i}$, and
    $y_i \in \mathcal{Y}_i \subset \mathbb{R}^{m_i}$,
where $J_i, R_i, G_i$ are the skewed interconnection, semi-positive definite dissipation matrices, and the
input map, respectively.
Each subsystem additionally exposes interconnection ports $w_i, z_i$ with unit port maps through which it exchanges power with neighboring subsystems.
By compositionality, two port-Hamiltonian systems interconnected through a power-preserving interconnection form a composite port-Hamiltonian system~\cite{van2014port}.
The dynamics of the interconnected system from \eqref{eq:pH-subsystems} is written as~\cite{neary2023compositional, van2025learning}:
\begin{equation}
    \begin{split}\label{eq:true-dynamics}
        \dot{x} &= \bigl[J(x) - R(x) + Q(x)\bigr]\nabla H(x)
              + G(x)u,\\
        y &= G^\top \nabla H(x),
    \end{split}
\end{equation}
where $J(x)$, $R(x)$, and $G(x)$ only contain block diagonal elements with $J(x) = {\rm diag}(\{J_i(x_i)\}_{i=1}^N)$ and similar structure for $R(x)$ and $G(x)$.
The input, state, and output are defined by concatenating the subsystems' input, state, and output, i.e., $x = {\rm vec}(\{x_i\}_{i=1}^N) \in \mathcal{D} \subset \mathbb{R}^{n}$ with $n = n_i + \dots + n_N$ and similarly for $u\in \mathcal{U}$ and $y\in \mathcal{Y}$.
The matrix $Q(x)$ encodes coupling information between the subsystems and is skew
\newgeometry{top=.75in, bottom=.75in, left=0.75in, right=0.75in}

\noindent symmetric with zero block-diagonal entries. 
In particular, $Q$ encodes which $z_i$ are coupled with which $w_j$ between two subsystems.


\vspace{0.5em}
\noindent\textit{Allowable set.}
The user prescribe a safety specification function
$\varphi: \mathcal D \to \mathbb{R}$
that defines the allowable set:
\[
    \mathcal{A}
    \coloneqq
    \bigl\{x \in \mathcal{D} : \varphi(x) \geq 0\bigr\},
\]
encoding safety requirements such as kinematic and charge constraints.
Our objective are:
\begin{enumerate}[(P1)]
    \item To determine 
    the relative degree of the safety specification $\varphi$
    with respect to the input placement in \eqref{eq:true-dynamics}.
    \item In the relative-degree-two case, we identify the minimal set of subsystems' local storage terms needed to reshape $\varphi$ into a relative-degree-one barrier, yielding a port-transversal barrier construction.
    \item Identify conditions under which the resulting port-transversal barrier satisfies the CBF feasibility condition under bounded inputs.
\end{enumerate}
The problems are resolved by the influence graph analysis and barrier synthesis of Section~\ref{sec:port-transversality}, using the structure $(J, R, Q, G)$.

\section{Main Results} \label{sec:main}
We introduce the influence graph-based relative degree analysis and present the synthesis of a port-transversal barrier.


\subsection{Port Transversality}\label{sec:port-transversality}

We begin by examining when a safety constraint is compatible with the input replacement of a port-Hamiltonian system. 
The key observation is that this compatibility is largely determined by the {interconnection topology} of the system.
In this section, we assume $\varphi$, $J$, $R$, $Q$, and $G$ to be $C^\infty$ smooth.

\subsubsection{Subsystem Interconnection Structure and Graphs}
Recall the composite port-Hamiltonian system in \eqref{eq:true-dynamics},
we write $A(x) \coloneqq J(x) - R(x) + Q(x)$
for the combined drift matrix.
The state $x$ decomposes into $N$ \emph{subsystems states}
$x_1, \dots, x_N$, 
reflecting the physical subsystems (inertias,
compliances, capacitances, inductances, etc.) constituting the network.
The key properties and assumptions of \eqref{eq:true-dynamics} are collected in the following.

\begin{property}[Separable storage]\label{ass:separable}
    The Hamiltonian decomposes as
      $H(x) = \sum_{i=1}^{N} H_i(x_i)$.
\end{property}

While the separable storage property comes for free under composition of subsystems, arbitrary state dependency of interconnection could complicate downstream analysis.

\begin{assumption}[Local dependence of interconnection]\label{ass:local-dep}
  For all subsystem indices $i,j$ and all $k \notin \{i,j\}$, $\frac{\partial Q_{ij}}{\partial x_k}(x) \equiv 0$.
\end{assumption}

In other words, block $Q_{ij}(x)$ only has state dependency on subsystem state $x_i$ or $x_j$.
Note that Assumption~\ref{ass:local-dep} is trivially satisfied when $Q$ is constant.
Furthermore, we assume local subsystem storage functions are strictly convex.
\begin{assumption}[Convexity]\label{ass:convexity}
  Each $H_i : \mathbb{R}^{n_i} \to \mathbb{R}$ is $C^2$ and strictly convex with a unique minimum at~$x_i^\star$.
\end{assumption}

While not intrinsic to the general framework, Assumption~\ref{ass:convexity} is physically natural: each storage element has a unique energy-minimizing equilibrium $x_i^\star$, as for typical springs, capacitors, and inertias. 
We now define two graphs on \eqref{eq:true-dynamics}.

\begin{definition}[Power flow graph]\label{def:power-flow-graph}
  The \emph{power flow graph} $\mathcal{G}(Q) = (\mathcal{V},
  \mathcal{E}_Q)$ has node set $\mathcal{V} = \{1, \dots, N\}$ and
  edge set
  $\mathcal{E}_Q = \bigl\{\{i,j\} : Q_{ij} \not=
  0\bigr\}$.
\end{definition}

\begin{definition}[Influence graph]\label{def:influence-graph}
  The \emph{influence graph} $\mathcal{G}(A) = (\mathcal{V},
  \mathcal{E}_A)$ has the same node set and edge set
  $\mathcal{E}_A = \bigl\{\{i,j\} : A_{ij}(x) \not\equiv
  0\bigr\}$.
\end{definition}

\begin{figure}[t]
    \centering
    \includegraphics[width=\linewidth]{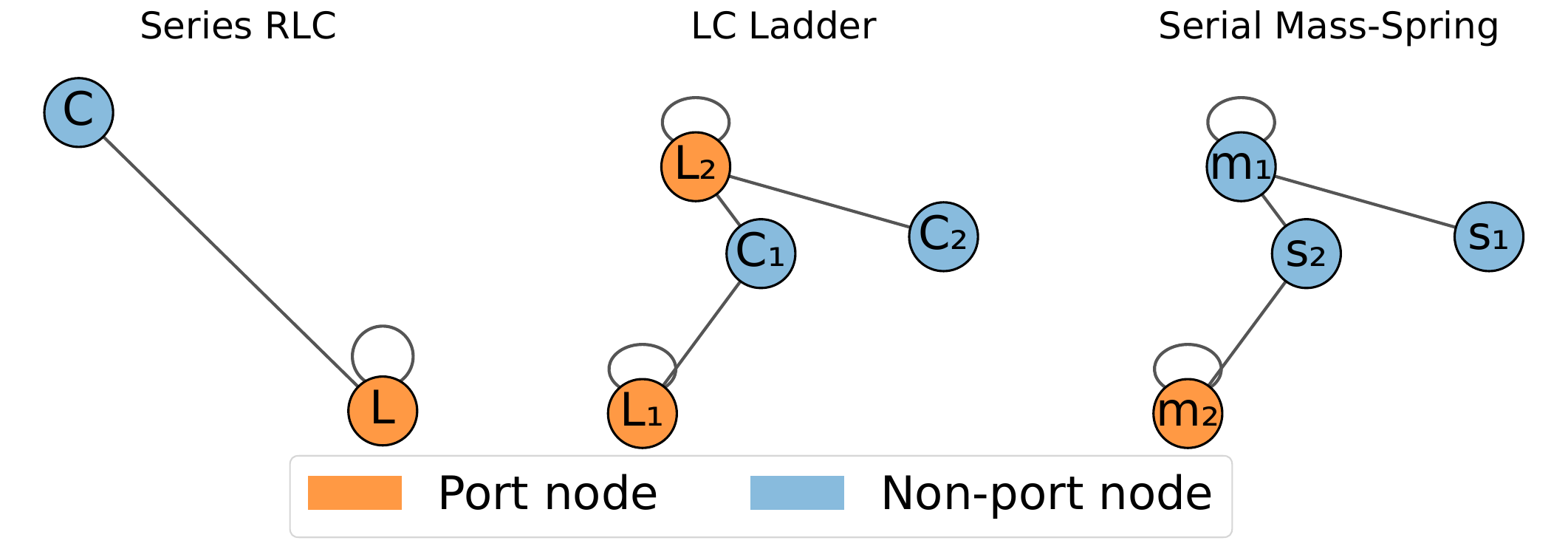}
    \caption{%
        Influence graphs $\mathcal{G}(A)$ of three port-Hamiltonian systems.
        %
        {(Left)} Series RLC:
        nodes $C$ and $L$ represent the capacitor and inductor subsystems states, and the edge $\{C, L\}$ encodes the power exchange between the capacitor and inductor.
        %
        {(Center)}
        LC Ladder with two voltage sources.
        %
        {(Right)} 
        Serial Mass-Spring:
        the port enters only at mass $m_2$ with potential and kinetic storage on springs $s_1, s_2$ and masses $m_1, m_2$ with input applying on the momentum of mass $m_2$.
        }
    \label{fig:influence_graphs}
    \vspace{-1.5em}
\end{figure}

It is worth noting that the power flow graph $\mathcal{G}(Q)$ retains physical interpretability in the sense that its edges represent conservative power exchange between subsystems.
In contrast, the influence graph $\mathcal{G}(A)$ governing the actual propagation of
the drift is the central object for our analysis.
Fig.~\ref{fig:influence_graphs} shows the influence graphs of three simple systems, illustrating that systems from different domains can share similar topological structure.
In addition, we introduce Example~\ref{ex:lc-ladder} to guide the discussion throughout this paper.

\begin{example}[Dual-inputs LC ladder network]\label{ex:lc-ladder}
Consider a two-section LC ladder driven by voltage sources
$u_1$ and $u_2$ through lossy inductors $L_1$ and $L_2$,
with capacitors $C_1$ and $C_2$ as shown in Fig.~\ref{fig:LC_ladder_diagram}.
The subsystem states are
$x = (q_1,\, q_2,\, \phi_1,\, \phi_2)^\top$,
where $q_i$ is the charge on~$C_i$ and $\phi_i$ the flux through~$L_i$.
The separable Hamiltonian is:
\[
  H(x)
  = \frac{q_1^2}{2C_1}
  + \frac{q_2^2}{2C_2}
  + \frac{\phi_1^2}{2L_1}
  + \frac{\phi_2^2}{2L_2},
\]
with efforts
$e \coloneqq \nabla H = (v_{C_1},\, v_{C_2},\, i_{L_1},\, i_{L_2})^\top$.
Notice that $z_i \coloneqq e_i$ for each subsystem, and
Kirchhoff's voltage and current laws yield the interconnection structure $Q$.
The structure tuple of \eqref{eq:true-dynamics} is then derived as $J = 0$, and:
\[
  R = \begin{psmallmatrix}
    0 & 0 & 0 & 0 \\
    0 & 0 & 0 & 0 \\
    0 & 0 & R_1 & 0 \\
    0 & 0 & 0 & R_2
  \end{psmallmatrix},\;
  Q = \begin{psmallmatrix}
    0 & 0 & 1 & -1 \\
    0 & 0 & 0 & 1 \\
    -1 & 0 & 0 & 0 \\
    1 & -1 & 0 & 0
  \end{psmallmatrix},\;
  G = \begin{psmallmatrix}
    0 & 0 \\ 0 & 0 \\ 1 & 0 \\ 0 & 1
  \end{psmallmatrix}.
\]
\noindent The corresponding influence graph $\mathcal G(A)$, illustrated in Fig.~\ref{fig:influence_graphs} (Center), has edges
$\{C_1, L_1\}$, $\{C_1, L_2\}$, $\{C_2, L_2\}$, $\{L_1, L_1\}$, and $\{L_2, L_2\}$ induced from the sparsity structure of $A$.
Lastly, we assign a constraint $q_1 \leq q_{\max}$ on $C_1$ to facilitate future discussion for barrier function synthesis.
\end{example}

\begin{figure}
    \centering
    \includegraphics[width=\linewidth]{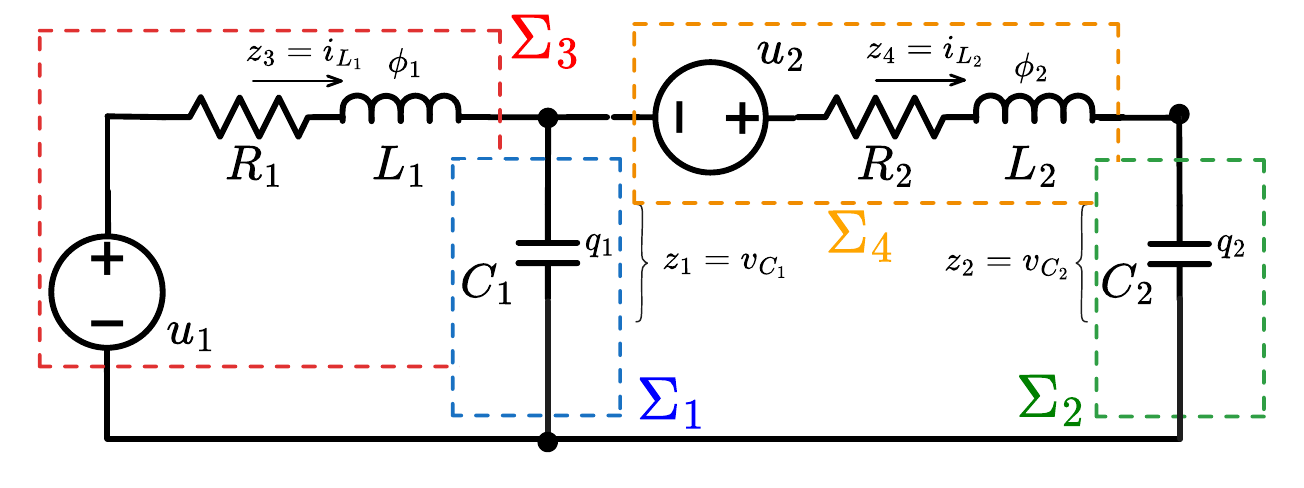}
    \caption{Dual-input LC ladder network that induces the influence graph in Fig.~\ref{fig:influence_graphs} (Center) with subsystems labeled as $\Sigma_i$.}
    \label{fig:LC_ladder_diagram}
    \vspace{-1.5em}
\end{figure}

\subsubsection{Port Nodes, Barrier Nodes, and Transversality}

Given a smooth safety specification $\varphi : \mathcal{D} \to \mathbb{R}$
defining the allowable set $\mathcal{A} = \{x \in \mathcal{D} : \varphi(x) \geq
0\}$ with $\nabla \varphi \neq 0$ on~$\partial\mathcal{A}$, we identify two
distinguished subsets of~$\mathcal{V}$:
\begin{align}
  \text{\emph{Port nodes:}} \quad
  \mathcal{P}
  &= \bigl\{i \in \mathcal{V} : G_i \not\equiv 0 \bigr\},
  \label{eq:port-nodes} \\
  \text{\emph{Barrier nodes:}} \quad
  \mathcal{B}
  &= \bigl\{i \in \mathcal{V} : \nabla_{x_i} \varphi \not\equiv 0 \bigr\},
  \label{eq:barrier-nodes}
\end{align}
where $G_i$ denotes the block of rows of~$G$ corresponding to
subsystem~$i$.

\begin{definition}[Subsystem support]\label{def:support}
  For a smooth scalar function $\varphi : \mathcal{D} \to \mathbb{R}$,
  the \emph{subsystem support} is
  $\operatorname{supp}(\varphi)
  \coloneqq \{i \in \mathcal{V} : \nabla_{x_i}\varphi \not\equiv
  0\}$.
\end{definition}

In this notation, $\mathcal{B} = \operatorname{supp}(\varphi)$ and
$\mathcal{P} = \{i : G_i \not\equiv 0\}$.
{In Example~\ref{ex:lc-ladder}, since we prescribe a charge constraint $q_1 \leq q_{\max}$ on $C_1$,
the barrier node set is $\mathcal{B} = \{C_1\}$.
The port node set is $\mathcal{P} = \{L_1, L_2\}$ because both inductors carry input ports.}

\begin{definition}[$k$-hop neighborhood]\label{def:k-hop}
  For a subset $\mathcal{W} \subseteq \mathcal{V}$, define
  $\mathcal{N}^0(\mathcal{W}) = \mathcal{W}$. 
  The $k$-hop neighborhood of $\mathcal{W}$ in the influence graph $\mathcal{G}(A)$ is defined recursively:
  \begin{equation*}
      \begin{aligned}
          &\mathcal{N}^{k+1}(\mathcal{W})
            = \\
            &\qquad\mathcal{N}^k(\mathcal{W})
              \;\cup\;
              \bigl\{j \in \mathcal{V}
                : \exists\, i \in \mathcal{N}^k(\mathcal{W})
                  \text{ with }
                  \{i,j\} \in \mathcal{E}_A
              \bigr\}.
      \end{aligned}
  \end{equation*}
\end{definition}
In this work, we treat \eqref{eq:true-dynamics} as a control-affine system 
  $\dot x = f(x) + G(x)\,u$ with 
  $G = [g_1 \;\cdots\; g_m]$, 
  and write 
  $L_G\varphi \coloneqq \nabla\varphi^\top G \in \mathbb{R}^{1\times m}$ 
  for any smooth scalar function~$\varphi$.
\begin{definition}[Relative degree]\label{def:relative-degree}
  We say $\varphi$ has \emph{relative degree}
  $r \in \mathbb{N}$ at $x_0 \in \mathcal{D}$ if there exists
  a neighborhood $U \subseteq \mathcal D$ s.t.:
  \[
    L_G L_f^k \varphi(x) = 0, 
    \quad \forall\, x \in U,
    \quad k = 0,\dots,r{-}2,
  \]
  and $L_G L_f^{r-1}\varphi(x_0) \neq 0$.
  If $L_G L_f^k \varphi(x) = 0$ for all $x \in U$ and 
  all $k \geq 0$, then $\varphi$ has 
  \emph{infinite relative degree} at~$x_0$.
\end{definition}

\begin{definition}[Port transversality]\label{def:port-transversality}
  The function~$\varphi$ is \emph{port-transversal} with respect to
  system~\eqref{eq:true-dynamics} if:
  \begin{equation}\label{eq:pt-condition}
    \nabla \varphi(x)^\top G(x) \neq 0,
    \quad \forall\,x \in \partial\mathcal{A}\setminus\mathcal Z,
  \end{equation}
  with $\mathcal Z$ a degeneracy set of measure zero with respect to $\partial\mathcal{A}$.
\end{definition}

In application, the degeneracy set $\mathcal Z$ is typically nonempty, and we will make the form of $\mathcal Z$ precise in
\eqref{eq:blanket-degeneracy}.




\subsubsection{Relative Degree from Graph Distance}
The goal of this section is to characterize the relative degree of $\varphi$ from the structure of $\mathcal G(A)$.


\begin{lemma}[Support propagation]\label{lem:propagation}
  In \eqref{eq:true-dynamics}, under Assumption~\ref{ass:local-dep},
  $\operatorname{supp}(L_f^k \varphi)
  \subseteq \mathcal{N}^k(\mathcal{B})$
  for all $k \geq 0$.
\end{lemma}


\begin{proof}
    First, we show that under Property~\ref{ass:separable} (separable storage) and Assumption~\ref{ass:local-dep} (local dependence), if $\operatorname{supp}(g) \subseteq \mathcal W$, then
    \begin{equation}\label{eq:supp-one-step}
        \operatorname{supp}(L_f g) \subseteq \mathcal N^1(\mathcal W)
    \end{equation}
    for any smooth $g: \mathcal D \to \mathbb R$ and base set $\mathcal W \subseteq \mathcal V$.
    We start by recognizing that the following equality holds:
    $$L_f g = \nabla g^\top f = \sum_{i\in \operatorname{supp}(g)} \!\!\!\!\nabla_{x_i} g(x)^\top \!\!\!\sum_{j: \{i, j\}\in \mathcal E_A}\!\!\!\!A_{ij}(x)e_j(x_j),$$
    where the outer sum holds by the definition of $\operatorname{supp}(\cdot)$, and the inner sum holds by Property~\ref{ass:separable} for $e_j(x_j) \coloneq \nabla H_j(x_j)$.
    Next, we check the variable dependence of each term in the sum: $\nabla g(x)$ depends only on $x_i$ where $i\in \operatorname{supp}(g)$ and $A_{ij}(x)$ depends only on $x_i$ and $x_j$ for $j$ in the edge $\{i, j\}\in \mathcal E_A$ by Assumption~\ref{ass:local-dep}.
    Therefore, $\operatorname{supp}(L_f g) \subseteq \operatorname{supp}(g)\,\cup\,\{j: \{i, j\}\in\mathcal E_A \text{ and } i \in \operatorname{supp}(g)\}$, which is, by Definition~\ref{def:k-hop}, a subset of $\mathcal N^1(\operatorname{supp}(g))$ and $\mathcal N^1(\mathcal W)$.
    Next, we notice that the base case $\operatorname{supp}(\varphi) = \mathcal B$ is true by Definition~\ref{def:support}.
    The inductive case is true by letting $\operatorname{supp}(L_f^k\varphi)\subseteq \mathcal N^k(\mathcal B)$ and substituting $g = L_f^k\varphi$ and $\mathcal W = \mathcal N^k(\mathcal B)$ in \eqref{eq:supp-one-step}.
    By the principle of mathematical induction, $\operatorname{supp}(L_f^k \varphi)
    \subseteq \mathcal{N}^k(\mathcal{B})$
    for all $k \geq 0$.
\end{proof}

Let $d \coloneqq \mathrm{dist}_{\mathcal{G}(A)}(\mathcal{B}, \mathcal{P}) = \min\{k: \mathcal{N}^k(\mathcal{B})\cap\mathcal{P}\neq \emptyset\}$ denote the shortest-path distance from~$\mathcal{B}$
to~$\mathcal{P}$ in the influence graph, with the convention $d = 0$
when $\mathcal{B} \cap \mathcal{P} \neq \emptyset$.
The following theorem is the key result for diagnosing the relative degree of a candidate barrier function $\varphi$ with respect to \eqref{eq:true-dynamics}.

\begin{theorem}[Relative degree from graph distance]%
  \label{thm:relative-degree}
  In \eqref{eq:true-dynamics}, under Assumption~\ref{ass:local-dep},
  consider a candidate barrier $\varphi: \mathcal D \to \mathbb R$:
  \begin{enumerate}[\upshape(i)]
    \item \label{item:rd-lower}
    \emph{Lower bound.}\;
        {For all $k < d$,\,
        $L_G L_f^k \varphi \equiv 0$ on~$\mathcal{D}$.
        Consequently, $\varphi$ has relative degree
        $r \geq d{+}1$ at every $x_0 \in \mathcal{D}$.}

    \item \label{item:rd-attain}
        \emph{Attainment (scalar subsystems).}\;
          Under Assumption~\ref{ass:convexity} (convexity),
          {suppose the structure matrices $J$, $R$, $Q$, $G$ are 
          constant, every subsystems are single state $n_i = 1$, $\forall i \leq N$, 
          and $\rho = (i_0, \dots, i_d)$ is the unique shortest path from~$\mathcal{B}$ to 
          $\mathcal{P}$, 
          %
          Then, $L_G L_f^d\varphi(x) \neq 0$ and the relative degree is
          $r = d + 1$.}
    
        
      \item \label{item:obstruction} 
          {\emph{Structural obstruction.}\;
          If no path exists between $\mathcal{B}$ and $\mathcal{P}$ in $\mathcal{G}(A)$, i.e., $d = \infty$, then 
          $L_G L_f^k \varphi \equiv 0$ for all $k \geq 0$.
          Moreover, 
          smooth static feedback that 
          $\kappa : \mathcal{D} 
          \to \mathbb{R}^m$ with modified drift 
          $\tilde{f} = f + G\kappa$,\;
          $L_G L_{\tilde{f}}^k \varphi \equiv 0,\; \forall k \geq 0$.
          }
  \end{enumerate}
\end{theorem}


\begin{proof}
    (\ref{item:rd-lower}): 
        By Lemma~\ref{lem:propagation},
        $\operatorname{supp}(L_f^k \varphi) \subseteq
        \mathcal{N}^k(\mathcal{B})$. If $k < d$, then
        $\mathcal{N}^k(\mathcal{B}) \cap \mathcal{P} = \emptyset$ by
        definition of shortest-path distance. 
        Hence, $\operatorname{supp}(L_f^k \varphi) \cap \mathcal{P} = \emptyset$ and
        $\nabla_{x_j}(L_f^k \varphi)
        \equiv 0$ for all $j \in \mathcal{P}$.
        Therefore, $L_G L_f^k \varphi = \sum_{j \in \mathcal{P}}
        \nabla_{x_j}(L_f^k \varphi)^\top G_j = 0$.

    {(\ref{item:rd-attain})}:
        The goal here is to show that under the three conditions: constant structure, scalar subsystems, and existence of unique shortest path, $L_GL_f^d \varphi \not = 0$.
        First, we write:
        \begin{equation}\label{eq:part2_LGLf_sum}
            L_GL_f^d \varphi = \nabla L_f^d\varphi^\top G = \sum_{j\in \mathcal P} \partial_{x_j}(L_f^d\varphi) G_j,
        \end{equation}
        which highlights that only $x_j$ with $j \in \mathcal P$ matters in $\nabla L_f^d\varphi$.
        Furthermore, since $\rho = (i_0, \dots, i_d)$ is the unique shortest path from~$\mathcal{B}$ to $\mathcal{P}$, equation~\eqref{eq:part2_LGLf_sum} reduces to:
        $$
            L_GL_f^d \varphi = \partial_{x_{i_d}}(L_f^d\varphi) G_{i_d}.
        $$
        Since $G_{i_d}$ is a scalar constant with $i_d \in \mathcal P$, $G_{i_d} \neq 0$.
        Therefore, showing $\partial_{x_{i_d}}(L_f^d\varphi) \neq 0$ suffices.
        Let $\varphi_k \coloneqq L_f^k\varphi$ for all $0 \leq k \leq d$. 
        By the product rule, $\partial_{x_{i_k}}\varphi_k$ can be expanded into:
        \begin{align*}
            \varphi_k &= \sum_{i \in \operatorname{supp}(\varphi_k)}\!\!\!\!\!\frac{\partial \varphi_{k-1}}{\partial x_i} \sum_{j:\{i, j\}\in \mathcal E_A}\!\!\!\!\!A_{ij}\frac{\partial H_j(x_j)}{\partial x_j}\\
            \frac{\partial \varphi_k}{\partial {x_{i_k}}} &= \sum_{i, j} \frac{\partial^2 \varphi_{k-1}}{\partial x_{i_k}\partial x_i} A_{ij}\nabla H_j(x_j) +\\
            &\qquad \sum_{i, j} \frac{\partial \varphi_{k-1}}{\partial x_i} \frac{\partial A_{ij}}{\partial {x_{i_k}}}\nabla H_j(x_j) +\\
            &\qquad \sum_{i, j} \frac{\partial \varphi_{k-1}}{\partial x_i} A_{ij}\frac{\partial^2 H_j(x_j)}{\partial x_{i_k}\partial x_j}
        \end{align*}
        where the first term is zero because $\frac{\partial \varphi_{k-1}}{\partial x_{i_k}} = 0$ by Lemma~\ref{lem:propagation}.
        The second term is zero by constant $A_{ij}$.
        By Property~\ref{ass:separable}, $\frac{\partial^2 H_j(x_j)}{\partial x_{i_k}\partial x_j} \neq 0$ iff $j = i_k$, hence, the last term and the expression reduces to 
        \begin{equation}\label{eq:simplified_sum_partial_varphi}
            \frac{\partial \varphi_k}{\partial {x_{i_k}}} = \sum_{i} \frac{\partial \varphi_{k-1}}{\partial x_i} A_{ii_k}
        \nabla^2_{x_{i_k}}\! H_{i_k}(x_{i_k}).
        \end{equation}
        To further simplify the R.H.S. in \eqref{eq:simplified_sum_partial_varphi}, we note that any non-zero $\partial_{x_i} \varphi_{k-1} A_{ii_k}
        \nabla^2_{x_{i_k}}\! H_{i_k}(x_{i_k})$ requires $A_{ii_k} \neq 0$ and $i \in \operatorname{supp}(\varphi_{k-1})$.
        To have $A_{ii_k} \neq 0$, we need $i$ to be adjacent to $i_k$, which forces $k-1 \leq \mathrm{dist}_{\mathcal{G}(A)}(\mathcal{B}, \{i\}) \leq k+1$.
        To have $i \in \operatorname{supp}(\varphi_{k-1})$, we need $\mathrm{dist}_{\mathcal{G}(A)}(\mathcal{B}, \{i\}) \leq k-1$.
        Together, it requires $\mathrm{dist}_{\mathcal{G}(A)}(\mathcal{B}, \{i\}) = k-1$.
        Moreover, since $i$ is adjacent to $i_k$ and $\rho = (i_0, \dots, i_d)$ is the unique shortest path between $\mathcal B$ and $\mathcal P$, $i=i_{k-1}$.
        Equation~\eqref{eq:simplified_sum_partial_varphi} is further simplified to:
        \begin{equation}\label{eq:simplified_partial_varphi}
            \frac{\partial \varphi_k}{\partial {x_{i_k}}} = \frac{\partial \varphi_{k-1}}{\partial x_{i_{k-1}}} A_{{i_{k-1}}i_k}\nabla^2_{x_{i_k}}\! H_{i_k}(x_{i_k}).
        \end{equation}
        Let $\psi_k \coloneqq \partial_{x_{i_k}}\varphi_k$ for $k = 0, \dots, d$, equation~\ref{eq:simplified_partial_varphi} can be expanded iteratively:
        \begin{align*}
            \psi_k  &= \psi_{k-1} A_{{i_{k-1}}i_k}\nabla^2_{x_{i_k}}\! H_{i_k}(x_{i_k})\\
                    &= \psi_{0} \prod_{l=1}^k A_{i_{l-1}i_l}\nabla_{x_{i_l}}^2 H(x_{i_l}).
        \end{align*}
        Set $k=d$, since each $A_{i_{l-1}i_l} \neq 0$ on the path $\rho$, and $\nabla_{x_{i_l}}^2 H(x_{i_l}) > 0$ by Assumption~\ref{ass:convexity}, $\psi_d = \partial_{x_{i_d}}(L_f^d\varphi) \neq 0$.

    (\ref{item:obstruction}): 
        If $\mathrm{dist}_{\mathcal{G}(A)}(\mathcal{B}, \mathcal{P}) = \infty$, let
        $\mathcal{C} = \bigcup_{k \geq 0}\mathcal{N}^k(\mathcal{B})$ be
        the set of all nodes reachable from~$\mathcal{B}$
        in~$\mathcal{G}(A)$. Then $d = \infty$ is equivalent to
        $\mathcal{C} \cap \mathcal{P} = \emptyset$.
        By Lemma~\ref{lem:propagation},
        $\operatorname{supp}(L_f^k \varphi) \subseteq
        \mathcal{N}^k(\mathcal{B}) \subseteq \mathcal{C}$, hence
        $\operatorname{supp}(L_f^k \varphi) \cap \mathcal{P} = \emptyset$ and
        thus $L_G L_f^k \varphi \equiv 0$ for all $k \geq 0$.
        
        For the memoryless feedback, consider any smooth static feedback
        $u = \kappa(x)$, giving $\tilde{f} = f + G\kappa$.
        Since
        $\mathcal{C} \cap \mathcal{P} = \emptyset$, every node
        $i \in \mathcal{C}$ satisfies $G_i \equiv 0$. 
        Hence, the feedback
        injection $G\kappa$ vanishes on~$\mathcal{C}$. We show
        $L_{\tilde{f}}^k\varphi = L_f^k \varphi$ for all $k\geq 0$ by induction. 
        The base case is trivial. 
        For the inductive case, assume $L_{\tilde{f}}^{k-1} \varphi = L_f^{k-1} \varphi$. 
        Then,
        \[
          L_{\tilde{f}}^k \varphi
          = \nabla(L_f^{k-1}\varphi)^\top (f + G\kappa)
          = L_f^k \varphi + \nabla(L_f^{k-1}\varphi)^\top G\kappa.
        \]
        By Lemma~\ref{lem:propagation},
        $\operatorname{supp}(L_f^{k-1}\varphi) \subseteq
        \mathcal{N}^{k-1}(\mathcal{B}) \subseteq \mathcal{C}$, and since
        $G_i \equiv 0$ on~$\mathcal{C}$, the term
        $\nabla(L_f^{k-1}\varphi)^\top G\kappa \equiv 0$. 
        Thus, $L_{\tilde{f}}^k \varphi = L_f^k \varphi$ for all $k \geq 0$, and consequently
        $L_G L_{\tilde{f}}^k \varphi = L_G L_f^k \varphi \equiv 0$ for all
        $k \geq 0$.
\end{proof}

Theorem~\ref{thm:relative-degree} identifies the relative degree from graph distance. 
We now propose a framework to reshape a high relative degree safety specification $\varphi$ into a port-transversal barrier informed by the influence graph structure.
\subsection{Port-Transversal Barrier Synthesis}

In this section, we consider the energy-shaping problem of a relative degree two safety specification $\varphi$.
The port-transversal barrier is a relative degree one candidate CBF that brings the problem into a form amenable to CBF feasibility analysis.

\subsubsection{Port-Energy Shaping of the Barrier}
For mechanical systems, the energy-aware barrier
\(
\varphi(q)-\tfrac{1}{\gamma}T(q,p)
\)
is well known for configuration constraints~\cite{singletary2021safety}, but its derivation relies on the canonical mechanical coordinates. 
We isolate the underlying principle as the \emph{barrier-insulating blanket} and generalize it beyond mechanics. 

\begin{definition}[Barrier-insulating blanket]%
\label{def:insulating-blanket}
Given the barrier node set
$\mathcal{B} = \operatorname{supp}(\varphi)$ and the influence graph
$\mathcal{G}(A) = (\mathcal{V}, \mathcal{E}_A)$, the
\emph{barrier-insulating blanket} is the vertex boundary:
\begin{equation}\label{eq:insulating-blanket}
\partial\mathcal{B}
\coloneqq
\bigl\{
  j \in \mathcal{V} \setminus \mathcal{B}
  :
  \exists\, i \in \mathcal{B}
  \text{ with }
  \{i,j\} \in \mathcal{E}_A
\bigr\}.
\end{equation}
We say that the barrier is \emph{port-insulated} if
$\partial\mathcal{B} \subseteq \mathcal{P}$.
\end{definition}

{In Example~\ref{ex:lc-ladder}, with barrier node set $\mathcal{B} = \{C_1\}$, the barrier-insulating blanket $\partial\mathcal{B} = \{L_1, L_2\} \subseteq \mathcal{P}$, so the constraint induced barrier $\varphi = q_{\max} - q_1$ is port-insulated.}
We note that by the construction of an influence graph, each subsystem dynamic decomposes as:
\begin{equation}\label{eq:barrier-flow-decomp}
\begin{split}
    \dot{x}_i\big|_{u=0}
&=
\underbrace{%
  \sum_{j \in \partial\mathcal{B}} Q_{ij}(x)\, e_j
}_{\text{blanket-mediated}}+
\underbrace{%
  \sum_{\substack{j \in \mathcal{B}\setminus \{i\}}}
  Q_{ij}\, e_j
  + A_{ii}(x) e_i
}_{\text{intra-barrier}},
\end{split}
\end{equation}
with $i \in \mathcal{B}$ and $e_j \coloneqq \nabla H_j(x_j)$ denoting the subsystem effort factor, where we have used
$A_{ij} = Q_{ij}$ for $i \neq j$ by zero block diagonal entries of $Q$ and $A_{ii} = J_{ii}-R_{ii}$
by block-diagonality of~$J$ and $R$.
No terms from $\mathcal{V} \setminus (\mathcal{B} \cup
\partial\mathcal{B})$ appear, since
$A_{ij} \equiv 0$ for those pairs.
Decomposition~\eqref{eq:barrier-flow-decomp} identifies
\(\partial\mathcal{B}\) as the source of the uncontrolled barrier
dynamics, and thus the natural set of port nodes to leverage for 
reshaping the barrier.

\begin{definition}[Port-transversal barrier]%
\label{def:pt-barrier-blanket}
Let $\varphi : \mathcal{D} \to \mathbb{R}$ be a function with subsystem support
$\operatorname{supp}(\varphi) = \mathcal{B}$ that admits an insulating blanket $\partial\mathcal{B}$ in the sense of
Definition~\ref{def:insulating-blanket}.
Suppose $\partial\mathcal{B} \cap \mathcal{P} \neq \emptyset$.
For weights $\beta_j > 0$ and shaping parameter $\gamma > 0$,
the \emph{port-transversal barrier} w.r.t. $\varphi$ is:
\begin{equation}\label{eq:pt-barrier-blanket}
h_\gamma(x)
\coloneqq
\varphi(x)
-
\frac{1}{\gamma}
\sum_{j \in \partial\mathcal{B} \cap \mathcal{P}}
\beta_j\,\bar{H}_j(x_j),
\end{equation}
where
$\bar{H}_j(x_j) \coloneqq H_j(x_j) - H_j(x_j^\star) \geq 0$
is the shifted local storage at the port node~$j$.
\end{definition}


{Recall from Assumption\ref{ass:convexity} that each $x_j^\star$ denotes the unique minimizer of the local storage~$H_j$.
Now, notice that since the input influences $h_\gamma$ through each local storage~$H_j$, the barrier loses first-order input sensitivity when every blanket effort $\nabla H_j(x_j)$ vanishes.
We collect these states into the \emph{blanket degeneracy set}:
\begin{equation}\label{eq:blanket-degeneracy}
\mathcal{Z}_{\partial\mathcal{B}}
\coloneqq
\{x \in \mathcal{D} : L_G h_\gamma(x) = 0\}.
\end{equation}
Managing CBF feasibility near this set is the main challenge addressed in Section~\ref{sec:pt-feasibility}.
Fig.~\ref{fig:series_rlc_pt_barrier} visualizes port-transversal barriers on a series RLC circuit.}

\begin{figure}[t]
    \centering
    \includegraphics[width=\linewidth]{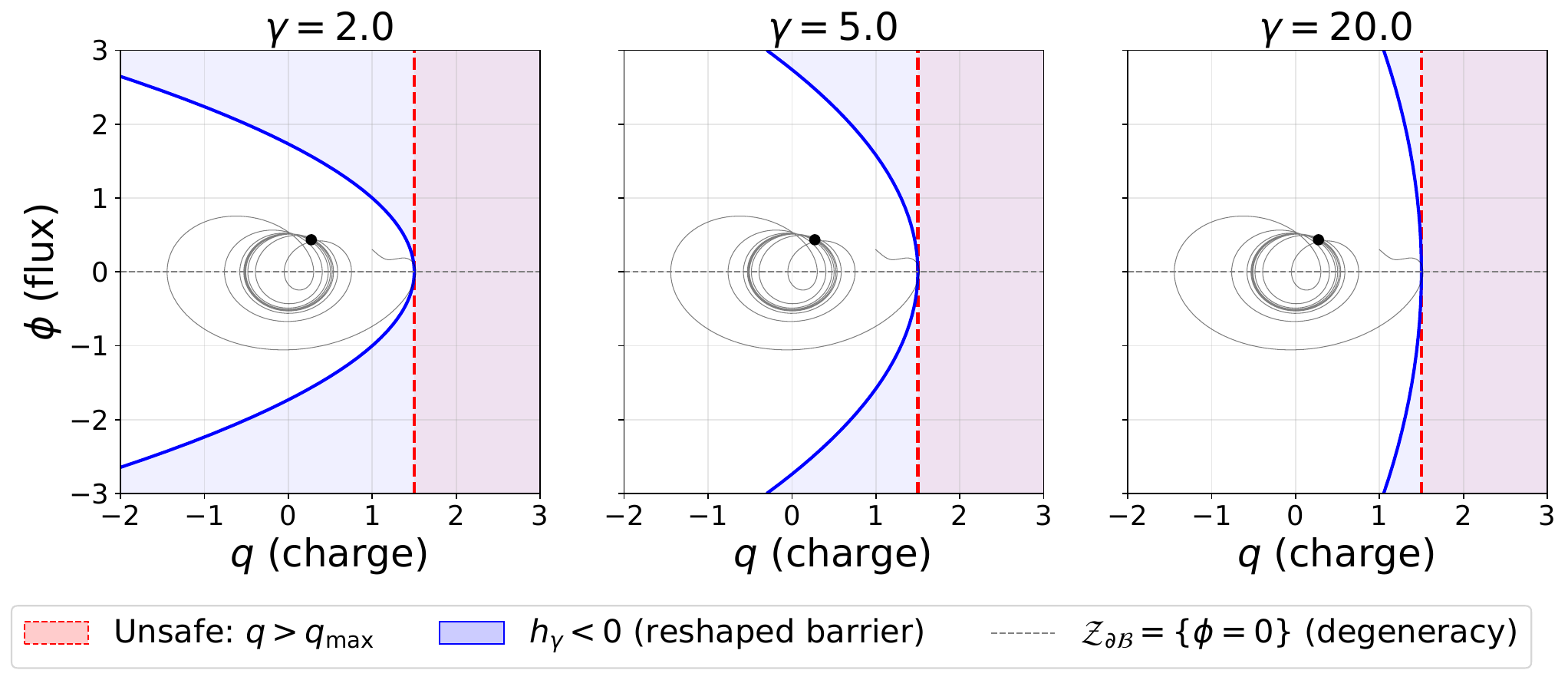}
    \caption{Port-transversal barrier for a charge constraint $q_{\max} = 1.5$ on the series RLC circuit.
        The safety specification $\varphi(x) = q_{\max} - q$ (dashed red) defines the allowable set $\mathcal{A} = \{q \leq q_{\max}\}$.
        The port-transversal barrier $h_\gamma(x) = \varphi(x) - \frac{1}{\gamma}\,\bar{H}_L(\phi)$ reshapes the safety specification by subtracting the shifted inductor storage $\bar{H}_L = \phi^2/(2L)$ from the insulating blanket
        $\partial\mathcal{B} = \{L\} = \mathcal{P}$ (blue contour).
        As $\gamma$ increases, the energy penalty shrinks and the safe set
        $\mathcal{S}_\gamma = \{h_\gamma \geq 0\}$ expands toward~$\mathcal{A}$,
        consistent with the monotonicity
        $\mathcal{S}_{\gamma_1} \subseteq \mathcal{S}_{\gamma_2}$ for
        $\gamma_1 \leq \gamma_2$ (Lemma~\ref{lem:pt-synthesis}\,(\ref{item:safe-set})).
        The black dashed line is the degeneracy set $\mathcal Z_{\partial \mathcal B}$, and
        the black trajectory shows the driven system under nominal input with the black dot the final state.
        }
    \label{fig:series_rlc_pt_barrier}
    \vspace{-1.5em}
\end{figure}

\begin{lemma}[Port-transversal barrier synthesis]%
\label{lem:pt-synthesis}
Let $\varphi : \mathcal{D} \to \mathbb{R}$ be a function with
$\operatorname{supp}(\varphi) = \mathcal{B}$ that admits $\partial\mathcal{B}$ in the sense of
Definition~\ref{def:insulating-blanket}.
Suppose $\partial\mathcal{B} \cap \mathcal{P} \neq \emptyset$,
and let $h_\gamma$ be the port-transversal barrier
of Definition~\ref{def:pt-barrier-blanket} with weights
$\{\beta_j\}_{j \in \partial\mathcal{B} \cap \mathcal{P}}$
and shaping parameter~$\gamma > 0$.
Then:
\begin{enumerate}[\upshape(i)]
  \item \label{item:pt-restored}
    \emph{Port transversality.}\;
    The input Lie derivative
    $L_G h_\gamma(x) \in \mathbb{R}^{1 \times m}$ satisfies:
    \begin{equation}\label{eq:Lg-hgamma}
      L_G h_\gamma(x)
      \;=\;
      -\frac{1}{\gamma}
      \sum_{j \in \partial\mathcal{B} \cap \mathcal{P}}
      \beta_j\,
      \nabla H_j(x_j)^\top G_j(x),
    \end{equation}
    which is nonzero for all
    $x \in \partial\mathcal{S}_\gamma \setminus
    \mathcal{Z}_{\partial\mathcal{B}}$.
    In particular, $h_\gamma$ has relative degree one on
    $\mathcal{D} \setminus \mathcal{Z}_{\partial\mathcal{B}}$.

    \item \label{item:safe-set}
      \emph{Safe-set inclusion.}\;
      $\mathcal{S}_\gamma \subseteq \mathcal{A}$.
      Moreover,
      $\mathcal{S}_{\gamma_1} \subseteq \mathcal{S}_{\gamma_2}$
      whenever $\gamma_1 \leq \gamma_2$, and
      \(
      \operatorname{int}(\mathcal{S})
      \;\subseteq\;
      \bigcup_{\gamma > 0}\mathcal{S}_\gamma
      \;\subseteq\;
      \mathcal{A}.
      \)

  \item \label{item:drift-decomp}
    \emph{Drift decomposition.}\;
    The drift Lie derivative decomposes as:
    \begin{equation}\label{eq:Lf-hgamma-decomp}
      L_f h_\gamma
      =
      L_f \varphi\big|_{\text{intra-}\mathcal{B}}
      + \Gamma_{\mathrm{blanket}}
      + \Gamma_{\mathrm{diss}}
      - \Gamma_{\mathrm{coupling}},
    \end{equation}
    where
    \begin{equation}\label{eq:drift-terms}
    \begin{aligned}
      L_f \varphi\big|_{\text{intra-}\mathcal{B}}
      &\coloneqq
      \sum_{i \in \mathcal{B}}
      \nabla_{x_i} \varphi^\top
      \Bigl(
        \sum_{\substack{j \in \mathcal{B}\setminus\{i\}}}
        Q_{ij}\, e_j+
        A_{ii}\, e_i
      \Bigr),
      \\
      \Gamma_{\mathrm{blanket}}
      &\coloneqq
      \sum_{i \in \mathcal{B}}
      \nabla_{x_i} \varphi^\top
      \sum_{j \in \partial\mathcal{B}}
      A_{ij}\, e_j,
      \\
      \Gamma_{\mathrm{diss}}
      &\coloneqq
      \frac{1}{\gamma}
      \sum_{j \in \partial\mathcal{B} \cap \mathcal{P}}
      \beta_j\, e_j^\top R_{jj}\, e_j
      \;\geq\; 0,
      \\
      \Gamma_{\mathrm{coupling}}
      &\coloneqq
      \frac{1}{\gamma}
      \sum_{j \in \partial\mathcal{B} \cap \mathcal{P}}
      \beta_j\, e_j^\top
      \sum_{k \neq j} A_{jk}\, e_k.
    \end{aligned}
    \end{equation}
\end{enumerate}
\end{lemma}


\begin{proof}
See Appendix~\ref{sec:proof-pt-synthesis}.
\end{proof}

{To illustrate the construction, consider again the LC ladder of Example~\ref{ex:lc-ladder} with barrier node set $\mathcal{B} = \{C_1\}$ and insulating blanket $\partial\mathcal{B} = \{L_1, L_2\} = \mathcal{P}$.
Applying Definition~\ref{def:pt-barrier-blanket} with equal weights $\beta_{L_1} = \beta_{L_2} = 1$, the port-transversal barrier is:
\[
  h_\gamma(x)
  = \underbrace{(q_{\max} - q_1)}_{\varphi(x)}
  - \frac{1}{\gamma}
    \Bigl(
      \underbrace{\frac{\phi_1^2}{2L_1}}_{\bar{H}_{L_1}(\phi_1)}
      +
      \underbrace{\frac{\phi_2^2}{2L_2}}_{\bar{H}_{L_2}(\phi_2)}
    \Bigr),
\]
where each subsystem storage $\bar{H}_{L_j} = \phi_j^2/(2L_j)$ penalizes magnetic energy in the corresponding inductor relative to its rest state $\phi_j^\star = 0$.
This example reflects the general principle: the influence graph identifies $\partial\mathcal{B}\cap \mathcal{P}$ as the minimal set of subsystems whose local storages must be leveraged in the reshaped barrier, independent of the physical domain.}

\subsubsection{Port-Transversal Barrier Feasibility}
\label{sec:pt-feasibility}

The synthesis of Lemma~\ref{lem:pt-synthesis} produces a barrier
of relative degree one with respect to \eqref{eq:true-dynamics} on
$\mathcal{D} \setminus \mathcal{Z}_{\partial\mathcal{B}}$, which is the first step to demonstrate its CBF feasibility.
It remains to be shown that the CBF inequality is feasible under bounded inputs.
We collect the required conditions into two groups:
port-side regularity and barrier-side regularity. 

\vspace{0.5em}
\noindent\textit{Port-Side Regularity.} Port-side regularity concerns the input structure at the blanket.
\begin{assumption}[Port-insulated barrier]%
\label{ass:port-insulated}
The barrier-insulating blanket satisfies
$\partial\mathcal{B} \subseteq \mathcal{P}$.
\end{assumption}

Evaluating the drift decomposition~\eqref{eq:Lf-hgamma-decomp} at states where every blanket effort vanishes, i.e., $e_j(x) = 0$ for all $j \in \partial\mathcal{B}$, 
the terms $\Gamma_{\mathrm{blanket}}$, $\Gamma_{\mathrm{diss}}$, and $\Gamma_{\mathrm{coupling}}$ vanish since each contains a blanket effort as a factor, leaving only the intra-barrier residual.
Since $\Gamma_{\mathrm{blanket}}$,
$\Gamma_{\mathrm{diss}}$, and $\Gamma_{\mathrm{coupling}}$ all
contain a blanket effort~$e_j$ as a factor and vanish when
$e_j = 0$ for every $j \in \partial\mathcal{B}$:
\begin{equation}\label{eq:residual-drift}
\begin{aligned}
    L_f h_\gamma(x)\big|_{\mathcal{Z}_{\partial\mathcal{B}}}
    =\sum_{i \in \mathcal{B}}
    \nabla_{x_i} \varphi(x)^\top
    \Bigl(
      \sum_{\substack{j \in \mathcal{B}\setminus \{i\}}}
      Q_{ij} e_j
      + A_{ii}(x) e_i
    \Bigr).
\end{aligned}
\end{equation}
\begin{assumption}[Blanket coercivity]%
\label{ass:blanket-coercivity}
Define the collective blanket effort norm
$\|e_{\partial\mathcal{B}}(x)\|
\coloneqq
\bigl(\sum_{j \in \partial\mathcal{B}}
\|e_j(x)\|^2\bigr)^{1/2}$.
There exists $\sigma_{\mathcal{B}} > 0$ such that:
\begin{equation}\label{eq:blanket-coercivity}
\|L_G h_\gamma(x)\|
\;\geq\;
\frac{\sigma_{\mathcal{B}}}{\gamma}\,
\|e_{\partial\mathcal{B}}(x)\|
\qquad \forall\, x \in \mathcal{D}.
\end{equation}
\end{assumption}

Assumption~\ref{ass:blanket-coercivity} ensures that the input
Lie derivative $L_G h_\gamma(x)$ can vanish only when every
blanket effort vanishes, and that away from this set the
control authority scales at least linearly with the collective
blanket effort.
Together with Assumption~\ref{ass:port-insulated}, the
degeneracy set reduces to:
\begin{equation}\label{eq:degen-reduced}
\mathcal{Z}_{\partial\mathcal{B}}
\!=\!
\bigl\{x\! \in \mathcal{D}\!
: e_j = 0,\forall j \in \partial\mathcal{B}\bigr\}
\!=\!
\bigl\{x\! : x_j = x_j^\star
,\forall j \in \partial\mathcal{B}\bigr\},
\end{equation}
where the first equality follows
from~\eqref{eq:blanket-coercivity} since
$\|L_G h_\gamma\| = 0$ forces
$\|e_{\partial\mathcal{B}}\| = 0$, and the second from strict
convexity of each~$H_j$ under
Assumption~\ref{ass:convexity}.
See~\cite[Appx.~VI-E]{leung2026port}
for further discussion on verifying blanket coercivity illustrated with Example~\ref{ex:lc-ladder}.

\vspace{0.5em}
\noindent\textit{Barrier-Side Regularity.}
In contrast, barrier-side regularity concerns the residual drift when input channels are silenced.

\begin{assumption}[Compact design safety boundary]%
\label{ass:compact}
The set $\partial\mathcal{S}_\gamma$ is compact.
\end{assumption}

Assumption~\ref{ass:compact} holds whenever the shifted storages $\bar{H}_j$ grow faster than~$\varphi$ along unbounded directions, which is the case, e.g., when each $H_j$ is strongly convex or has superlinear growth.
It fails if the storages grow only linearly and the
allowable set~$\mathcal{A}$ is unbounded.


\begin{assumption}[Dominated degeneracy]%
\label{ass:dom-degen}
There exists an open neighborhood~$U$ of
$\mathcal{Z}_{\partial\mathcal{B}} \cap
\partial\mathcal{S}_\gamma$
in~$\partial\mathcal{S}_\gamma$ and $c_0 \geq 0$ such that:
\begin{equation}\label{eq:dom-degen}
(-L_f h_\gamma(x))^+ \leq c_0 \|e_{\partial \mathcal B}(x)\|
\qquad \forall\, x \in U.
\end{equation}
\end{assumption}

Dominated degeneracy asserts that adversarial drifts, i.e., $f(x)$ such that $\nabla h_\gamma^\top f(x) < 0$, die out faster than the storage effort $e_{\partial \mathcal B}$ around the degeneracy set near the safety boundary.
With these assumptions in place, we construct the bounded input feasibility conditions.





\begin{lemma}[Drift envelope on the design boundary]%
    \label{lem:adverse-drift-envelope}
    Under Assumptions~\ref{ass:local-dep} (local dependence), \ref{ass:port-insulated} (port-insulated barrier), and
    \ref{ass:compact} (compact design safety boundary),
    there exist constants $M,c_1>0$ such that,
    \begin{equation}\label{eq:adverse-drift-envelope}
    (-L_f h_\gamma(x))^+
    \le
    M + c_1\,\|e_{\partial\mathcal B}(x)\|,
    \end{equation}
    for all $x\in \partial\mathcal S_\gamma$.
\end{lemma}

\begin{proof}
    See Appendix~\ref{sec:proof-adverse-drift-envelope}.
\end{proof}

With Lemma~\ref{lem:adverse-drift-envelope} in hand, we are ready to show CBF feasibility under bounded inputs.
\begin{theorem}[CBF feasibility under bounded inputs]\label{thm:pt-feasibility}
    Under Assumptions~\ref{ass:local-dep} (local dependence), \ref{ass:port-insulated} (port-insulated barrier),
    \ref{ass:blanket-coercivity} (blanket coercivity),
    \ref{ass:dom-degen} (dominated degeneracy),
    and~\ref{ass:compact} (compact design safety boundary),
    the worst-case required authority ratio $\bar{u}_\gamma$ is finite:
    \begin{equation}\label{eq:authority-bound}
    \bar{u}_\gamma
    \coloneqq
    \sup_{x \in \partial\mathcal{S}_\gamma
          \setminus \mathcal{Z}_{\partial\mathcal{B}}}
    \frac{\bigl(-L_f h_\gamma(x)\bigr)^+}
         {\|L_G h_\gamma(x)\|}
    < \infty.
    \end{equation}
    Consequently, if the admissible input set satisfies:
    \[
    \mathcal U \supseteq \{u:\|u\|\le \bar u_\gamma\},
    \]
    then the CBF condition for $h_\gamma$ is feasible on
    $\partial \mathcal S_\gamma$. Hence $\mathcal S_\gamma$ is
    forward invariant under any safety filter that enforces the CBF
    constraint.
\end{theorem}


\begin{proof}
    Define
        $\rho_\gamma(x)
        \coloneqq
        \frac{(-L_f h_\gamma(x))^+}{\|L_G h_\gamma(x)\|}$, with
        $x\in \partial\mathcal S_\gamma\setminus \mathcal Z_{\partial\mathcal B}.$
    By Assumption~\ref{ass:dom-degen} and Assumption~\ref{ass:blanket-coercivity}, there exists an open
    neighborhood $U$ of
    $\mathcal Z_{\partial\mathcal B}\cap\partial\mathcal S_\gamma$
    in $\partial\mathcal S_\gamma$ such that
    $$\rho_\gamma \leq \frac{c_0\gamma}{\sigma_{\mathcal B}}$$
    for some $c_0, \sigma_{\mathcal B} \geq 0$ on $U$.
    Next, if $\|e_{\partial\mathcal B}(x)\|=0$, then $e_j(x)=0$ for every
    $j\in\partial\mathcal B$ and in particular for every
    $j\in\partial\mathcal B\cap\mathcal P$.
    Substituting $e_j(x)=0$ into the
    equation~\eqref{eq:Lg-hgamma} kills every summand, giving
    $L_G h_\gamma(x)=0$, i.e.,
    $x\in\mathcal Z_{\partial\mathcal B}$ by~\eqref{eq:blanket-degeneracy}.
    Hence every zero-effort point $\{x\in \partial\mathcal S_\gamma\!: \|e_{\partial\mathcal B}(x)\|=0\}$ of $\partial\mathcal S_\gamma$ lies in
    $\mathcal Z_{\partial\mathcal B}\cap\partial\mathcal S_\gamma$,
    which is contained in $U$ by its definition in Assumption~\ref{ass:dom-degen}.
    Therefore, $\|e_{\partial\mathcal B}\|>0$ on the compact set
    $\partial\mathcal S_\gamma\setminus U$, and by continuity there
    exists $\delta>0$, such that,
    \[
        \|e_{\partial\mathcal B}(x)\|\ge \delta,
        \qquad
        x\in \partial\mathcal S_\gamma\setminus U.
    \]
    By Lemma~\ref{lem:adverse-drift-envelope},
    \[
    (-L_f h_\gamma)^+
    \le
    M + c_1\,\|e_{\partial\mathcal B}\|
    \qquad
    \text{on }\partial\mathcal S_\gamma.
    \]
    By Assumption~\ref{ass:blanket-coercivity},
    \[
    \|L_G h_\gamma\|
    \ge
    (\sigma_{\mathcal B}/\gamma)\,\|e_{\partial\mathcal B}\|,
    \]
    so for every
    $x \in \partial\mathcal S_\gamma \setminus U$,
    \begin{equation}\label{eq:thm:pt-feasibility:auth-ratio-bound}
        \rho_\gamma(x) = \frac{(-L_f h_\gamma)^+}
             {\|L_G h_\gamma\|}
        \le
        \frac{\gamma\,c_1}{\sigma_{\mathcal B}}
        +
        \frac{\gamma\,M}{\sigma_{\mathcal B}\,\delta}.
    \end{equation}
    Combining \eqref{eq:thm:pt-feasibility:auth-ratio-bound} with $\rho_\gamma \leq \frac{c_0 \sigma_{\mathcal B}}{\gamma}$ on
    $U\setminus \mathcal Z_{\partial\mathcal B}$ shows that
    $\rho_\gamma$ is uniformly bounded below by $\max\{\frac{c_0\gamma}{\sigma_{\mathcal B}}, \frac{\gamma c_1}{\sigma_{\mathcal B}}
        +
        \frac{\gamma M}{\sigma_{\mathcal B}\delta}\}$ on
    $\partial\mathcal S_\gamma\setminus \mathcal Z_{\partial\mathcal B}$,
    and therefore $\bar u_\gamma<\infty$.

    It remains to connect $\bar{u}_\gamma < \infty$ to CBF feasibility.
    On $\partial\mathcal{S}_\gamma$, $h_\gamma(x) = 0$
    and $\alpha(0) = 0$ for class-$\mathcal K$ function $\alpha$, so the CBF condition reduces to finding $u \in \mathcal{U} \supseteq \{u:\|u\|\le u_{\max}\}$ with
    $L_G h_\gamma(x)\, u \geq -L_f h_\gamma(x)$.
    If $L_f h_\gamma(x) \geq 0$ the drift is already
    nonnegative, and $u = 0$ is feasible.
    If $L_f h_\gamma(x) < 0$, the input must compensate
    for the adverse drift $(-L_f h_\gamma(x))^+ = |L_f h_\gamma(x)|$.
    The candidate minimum-norm input that could achieve $L_G h_\gamma(x)\, u \geq (-L_f h_\gamma(x))^+$ is
    $u^\star
    = \|u^\star\|
      \frac{L_G h_\gamma^\top}{\|L_G h_\gamma\|}$, with equality attained at $\|u^\star\| = \rho_\gamma$.
    So, feasibility in this case is equivalent to
    $\|u^\star\| = \rho_\gamma(x) \leq u_{\max}$.
    Since $\bar u_\gamma < \infty$, any $u_{\max} \geq \bar{u}_\gamma$ suffices
    on $\partial\mathcal{S}_\gamma \setminus
    \mathcal{Z}_{\partial\mathcal{B}}$.
    
    Any input bound
    $\mathcal{U} \supseteq \{u : \|u\| \leq \bar{u}_\gamma\}$
    makes the CBF condition feasible on all of
    $\partial\mathcal{S}_\gamma$.
    Forward invariance of $\mathcal{S}_\gamma$ then follows
    from the Nagumo-type result
    of~\cite[Thm.~2]{ames2019control}.
\end{proof}

Violation of dominated degeneracy in an underactuated mechanical system renders the corresponding energy-aware CBF infeasible under bounded inputs.
The serial spring--mass system in Fig.~\ref{fig:influence_graphs}~(Right) provides one such example.
This observation motivates us to look for structural conditions under which the assumptions hold automatically.

\subsubsection{Strong Structural Condition}
\label{sec:strong-struct}
In the following, we discuss a structural condition under which \eqref{eq:true-dynamics} automatically satisfies Assumption~\ref{ass:dom-degen}.

\begin{definition}[Strong structural condition]%
\label{def:strong-structural}
{A barrier set $\mathcal{B}$ satisfies the
\emph{strong structural condition} if the subgraph 
of~$\mathcal{G}(A)$ induced on~$\mathcal{B}$ has no 
edges (including self-edges).}
\end{definition}
Definition~\ref{def:strong-structural} states that no energy is exchanged or dissipated within the barrier subsystem set and that every term in the drift that acts on barrier states is mediated by a blanket effort.
The condition holds for mechanical configuration constraints, where $\mathcal{B} = \{q\}$ with no intra-$\mathcal{B}$ coupling and no configuration-level dissipation.
The condition also holds for constraints on conservative scalar subsystems, such as the capacitor constraint in Example~\ref{ex:lc-ladder}, where $\mathcal{B} = \{C_1\}$ is a single lossless storage element.
The strong structural condition eliminates the need for the neighborhood argument in Theorem~\ref{thm:pt-feasibility} and yields a simplified authority bound.

\begin{corollary}[Feasibility under strong structural condition]%
\label{cor:strong-feasibility}
Suppose
Assumptions
\ref{ass:local-dep},
\ref{ass:port-insulated},
\ref{ass:blanket-coercivity},
and~\ref{ass:compact} hold, and the barrier
set~$\mathcal{B}$ satisfies the strong structural condition
(Definition~\ref{def:strong-structural}).
Then the authority bound~\eqref{eq:authority-bound} simplifies
to:
\begin{equation}\label{eq:authority-strong}
  \bar{u}_\gamma
  \;\leq\;
  \frac{\gamma\, c_1}{\sigma_{\mathcal{B}}}.
\end{equation}
\end{corollary}
\begin{proof}
See Appendix~\ref{sec:proof-strong-feasibility}.
\end{proof}

\section{Numerical Simulations}\label{sc:experiment}
\vspace{-.25em}
We validate the structural predictions of the port-transversal barrier framework
on the LC ladder network of Example~\ref{ex:lc-ladder}.
The system parameters are
$C_1 = 1.0$, $C_2 = 2.0$, $L_1 = 0.5$, $L_2 = 1.0$,
$R_1 = 0.1$, $R_2 = 0.05$,
with charge constraint $q_1 \leq q_{\max} = 0.8$ on capacitor~$C_1$.
The barrier node set is $\mathcal{B} = \{C_1\}$ and the
barrier-insulating blanket is
$\partial\mathcal{B} = \{L_1, L_2\}$.
We compare two input configurations sharing the same physics
($J$, $R$, $Q$, $H$) but differing only in the input map~$G$:
\begin{center}
\scalebox{0.92}{%
\begin{tabular}{@{}lcc@{}}
  & Single input ($m\!=\!1$) & Dual input ($m\!=\!2$) \\
  \hline
  $\mathcal{P}$
    & $\{L_1\}$ & $\{L_1,\, L_2\}$ \\
  $\partial\mathcal{B} \cap \mathcal{P}$
    & $\{L_1\}$ & $\{L_1,\, L_2\}$ \\
  Port-insulated
    & No\; ($L_2 \notin \mathcal{P}$)
    & Yes \\
  Thm.~\ref{thm:pt-feasibility}
    & Does not apply
    & $\bar{u}_\gamma < \infty$
\end{tabular}}
\end{center}
 
For the dual-input systems, the port-transversal barrier is
$h_\gamma(x) = (q_{\max} - q_1)
  - \frac{1}{\gamma}\bigl(\bar{H}_{L_1}(\phi_1) + \bar{H}_{L_2}(\phi_2)\bigr)$
with $\bar{H}_{L_j} = \phi_j^2/(2L_j)$, which includes both blanket
storages.
For the single-input system, only $L_1$ carries a port, so the barrier only subtracts
$\bar{H}_{L_1}$:
$h_\gamma(x) = (q_{\max} - q_1)
  - \frac{1}{\gamma}\bar{H}_{L_1}(\phi_1)$.
We note that using the full-blanket barrier $\bar{H}_{L_1} + \bar{H}_{L_2}$ for the single-input system does not resolve the infeasibility in the following simulations, since $G_{L_2} = 0$ leaves $L_G h_\gamma$ unchanged while adding uncontrollable drift terms.

We set $\gamma = 4.0$ and $\alpha(s) = 1.5\, s$.
Lastly, both systems are driven by a nominal sinusoidal input and filtered by
the standard CBF-QP safety filter with input bound
$\|u\| \leq u_{\max} = 10$, using forward Euler integration
($\Delta t = 10^{-3}$, $T = 25\,\mathrm{s}$) from a shared initial
condition $x_0 = (0.6,\, 0.2,\, 0.5,\, -0.3)^\top$.

\subsection{Constraint Enforcement}
 
 
Fig.~\ref{fig:lc_authority_phase} reveals the mechanism underlying this failure.
On the left panel, the control authority $\|L_G h_\gamma(x)\|$ over time shows that 
the dual-input barrier maintains
positive authority throughout,
whereas the single-input authority occasionally vanishes.
Notice that the single-input trajectory crosses the
constraint line at $\phi_1 = 0$, which is within the open neighborhood of the degeneracy set, as predicted by Theorem~\ref{thm:pt-feasibility}.

%
\begin{figure}[t]
    \centering
    \includegraphics[width=\linewidth]{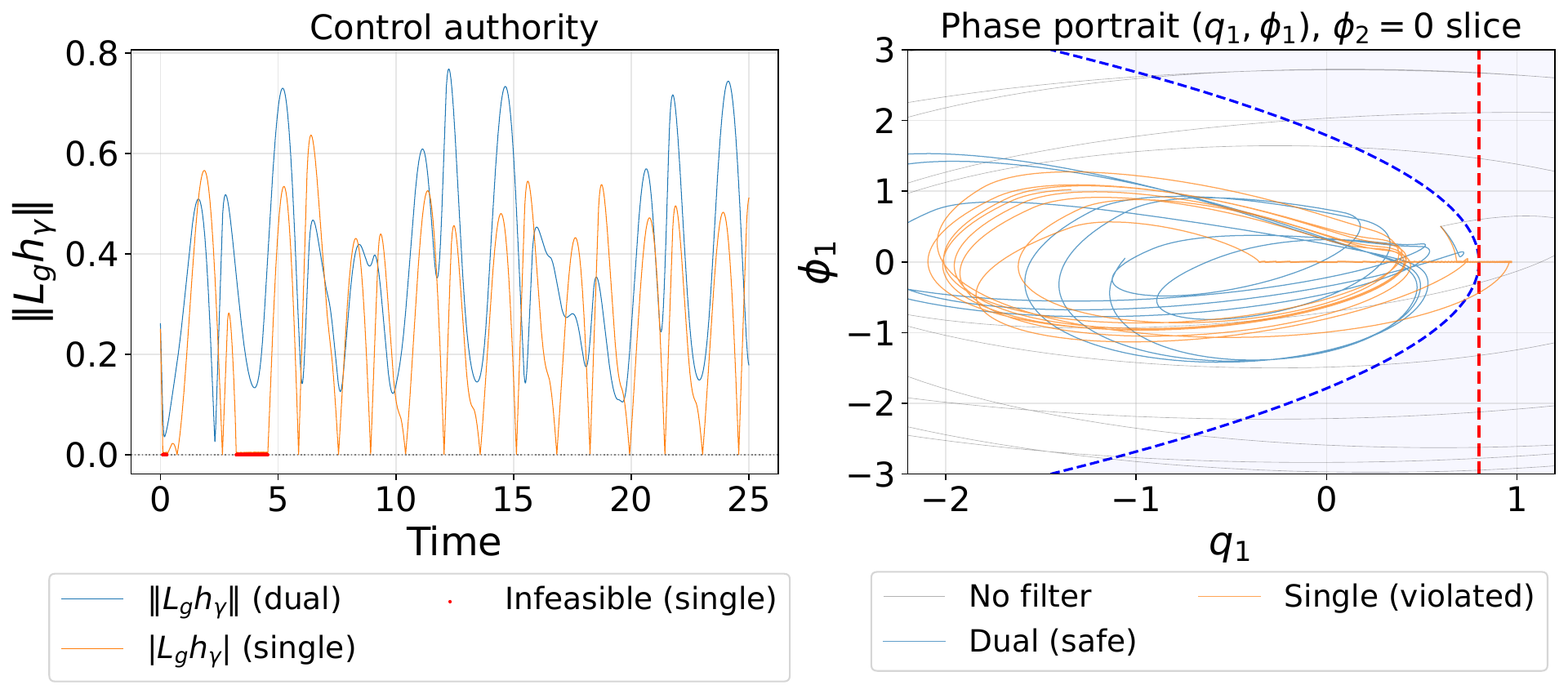}
    \caption{Control authority and phase portrait.
      (Left)~The control authority $\|L_G h_\gamma\|$: the
      dual-input system never loses input sensitivity, while
      the single-input authority drops to zero
      (red dots mark infeasible steps).
      (Right)~Phase portrait $(q_1,\phi_1)$ at the $\phi_2 = 0$
      slice: the dual-input trajectory (blue) stays within
      the safe set; the single-input trajectory (orange)
      crosses the constraint boundary.}
    \label{fig:lc_authority_phase}
    \vspace{-1em}
\end{figure}

\section{Conclusion}
We developed a graph-theoretic framework for safety-critical control of composite port-Hamiltonian systems.
The influence graph, constructed from the interconnection graph structure of the underlying network of subsystems, 
yields a lower bound on the relative degree of any safety constraint via shortest-path distance and identifies structural obstructions when no path to the input ports exists (Theorem~\ref{thm:relative-degree}).
When the relative degree exceeds one, the barrier-insulating blanket prescribes which subsystems states must enter a reshaped barrier.
The resulting port-transversal barrier (Definition~\ref{def:pt-barrier-blanket}) restores relative degree one, and its CBF feasibility under bounded inputs follows from three verifiable conditions together with a compactness assumption: (i) port-insulation, (ii) blanket coercivity, and (iii) benign degeneracy (Theorem~\ref{thm:pt-feasibility}).
Simulations on the LC ladder network confirm that the graph-theoretic condition $\partial\mathcal{B} \subseteq \mathcal{P}$ differentiates enforceable from unenforceable safety filtering.

\bibliographystyle{IEEEtran}
\bibliography{refs}

\section{Appendix}\label{sec:appendix}
Additional conceptual constructs are organized here.

\subsection{Proof of Lemma~\ref{lem:pt-synthesis} (Port-transversal barrier synthesis)}
\label{sec:proof-pt-synthesis}

\begin{proof}
(\ref{item:pt-restored}):\;
Since
$\operatorname{csupp}(\varphi) = \mathcal{B}$ and
$\partial\mathcal{B} \cap \mathcal{B} = \emptyset$,
we have $\nabla_{x_j} \varphi \equiv 0$ for every
$j \in \partial\mathcal{B}$.
The only terms in~\eqref{eq:pt-barrier-blanket} that depend
on~$x_j$ for $j \in \partial\mathcal{B} \cap \mathcal{P}$
are the shifted storages~$\bar{H}_j$, so
$\nabla_{x_j} h_\gamma
= -(\beta_j/\gamma)\,\nabla H_j(x_j)$.
For $\ell \in \mathcal{P} \setminus
(\partial\mathcal{B} \cap \mathcal{P})$, either
$\ell \notin \partial\mathcal{B}$ (so no shifted storage
term involves~$x_\ell$) or $\ell \in \mathcal{B}$
(excluded by $\partial\mathcal{B} \cap \mathcal{B}
= \emptyset$), giving
$\nabla_{x_\ell} h_\gamma^\top G_\ell = 0$ or
a term involving $\nabla_{x_\ell} \varphi$ with
$G_\ell \equiv 0$ (for $\ell \notin \mathcal{P}$).
Summing over all port nodes:
\begin{align*}
  L_G h_\gamma
  &= \sum_{\ell \in \mathcal{P}}
     \nabla_{x_\ell} h_\gamma^\top G_\ell
  =
  \sum_{j \in \partial\mathcal{B} \cap \mathcal{P}}
  \nabla_{x_j} h_\gamma^\top G_j
  \\
  &=
  -\frac{1}{\gamma}
  \sum_{j \in \partial\mathcal{B} \cap \mathcal{P}}
  \beta_j
  \nabla H_j(x_j)^\top G_j(x),
\end{align*}
yielding~\eqref{eq:Lg-hgamma}.
By~\eqref{eq:blanket-degeneracy},
$L_G h_\gamma(x) = 0$ if and only if
$x \in \mathcal{Z}_{\partial\mathcal{B}}$.

(\ref{item:safe-set}):\;
Strict convexity
in Assumption~\ref{ass:convexity} gives
$\bar{H}_j \geq 0$ for each
$j \in \partial\mathcal{B} \cap \mathcal{P}$,
so $h_\gamma \leq \varphi$ pointwise, hence
$\mathcal{S}_\gamma \subseteq \mathcal{A}$.
Since the penalty
$\frac{1}{\gamma}\sum_{j \in \partial\mathcal{B}\cap\mathcal{P}}
\beta_j \bar{H}_j$
is nonincreasing in~$\gamma$, the monotonicity
$\mathcal{S}_{\gamma_1} \subseteq \mathcal{S}_{\gamma_2}$
for $\gamma_1 \leq \gamma_2$ follows.
Now let $x \in \operatorname{int}(\mathcal{A})$, so $\varphi(x) > 0$.
Define
\[
P(x)
\;\coloneqq\;
\sum_{j \in \partial\mathcal{B}\cap\mathcal{P}}
\beta_j \bar{H}_j(x_j) \;\geq\; 0.
\]
If $P(x)=0$, then $h_\gamma(x)=\varphi(x)>0$ for every $\gamma>0$.
If $P(x)>0$, choose any
$\gamma > P(x)/\varphi(x)$.
Then
\(
h_\gamma(x)
=
\varphi(x) - \frac{1}{\gamma}P(x)
> 0,
\)
so $x \in \mathcal{S}_\gamma$.
Hence
$\operatorname{int}(\mathcal{A})
\subseteq \bigcup_{\gamma>0}\mathcal{S}_\gamma$.
The reverse inclusion follows from
$\mathcal{S}_\gamma \subseteq \mathcal{A}$ for every $\gamma>0$.

(\ref{item:drift-decomp}):\;
Expanding
$L_f h_\gamma = L_f \varphi
- \frac{1}{\gamma}
\sum_{j \in \partial\mathcal{B} \cap \mathcal{P}}
\beta_j\, L_f \bar{H}_j$
and using~\eqref{eq:barrier-flow-decomp} to split~$L_f \varphi$
gives the intra-$\mathcal{B}$ and
$\Gamma_{\mathrm{blanket}}$ terms.
For each shifted storage,
$L_f \bar{H}_j
= e_j^\top \sum_{k} A_{jk}\, e_k
= -e_j^\top R_{jj}\, e_j
  + e_j^\top \sum_{k \neq j} A_{jk}\, e_k$,
yielding $\Gamma_{\mathrm{diss}}$ and
$\Gamma_{\mathrm{coupling}}$ after sign reversal.
Since $R_{jj} \succeq 0$, we have
$\Gamma_{\mathrm{diss}} \geq 0$.
\end{proof}

\subsection{Proof of Lemma~\ref{lem:adverse-drift-envelope} (Drift envelope on the design boundary)}
\label{sec:proof-adverse-drift-envelope}

\begin{proof}
    From the drift decomposition~\eqref{eq:Lf-hgamma-decomp},
    \[
    -L_f h_\gamma
    =
    -L_f \varphi|_{\text{intra-}\mathcal B}
    -\Gamma_{\mathrm{blanket}}
    -\Gamma_{\mathrm{diss}}
    +\Gamma_{\mathrm{coupling}}.
    \]
    Since $\Gamma_{\mathrm{diss}}\ge 0$ by definition of $R$, we obtain
    \begin{align*}
        -L_f h_\gamma
        &\le
        -L_f \varphi|_{\text{intra-}\mathcal B}
        -\Gamma_{\mathrm{blanket}}
        +\Gamma_{\mathrm{coupling}},\\
        -L_f h_\gamma
        &\le
        \bigl|L_f \varphi|_{\text{intra-}\mathcal B}\bigr|
        +
        |\Gamma_{\mathrm{blanket}}|
        +
        |\Gamma_{\mathrm{coupling}}|,
    \end{align*}
    because the R.H.S. is non-negative, we have
    \[
    (-L_f h_\gamma)^+
    \le
    \bigl|L_f \varphi|_{\text{intra-}\mathcal B}\bigr|
    +
    |\Gamma_{\mathrm{blanket}}|
    +
    |\Gamma_{\mathrm{coupling}}|.
    \]

    Moving forward, we want to show that both $\Gamma_{\text{blanket}}$ and $\Gamma_{\text{coupling}}$ can be written as sums of the form $\sum_{j \in \partial\mathcal{B} \cap \mathcal{P}} e_j^\top c$.
    For $\Gamma_{\mathrm{coupling}}$, the outer factor $e_j$ with
    $j \in \partial\mathcal{B} \cap \mathcal{P}$ is already
    present in~\eqref{eq:drift-terms}.
    Absorbing the remaining factors into a coefficient function
    $\phi_j(x)
    \coloneqq
    (\beta_j/\gamma)\sum_{k \neq j} A_{jk}(x)\,e_k(x)
    \in \mathbb{R}^{n_j}$
    gives
    $\Gamma_{\mathrm{coupling}}
    =
    \sum_{j \in \partial\mathcal{B} \cap \mathcal{P}}
    e_j^\top \phi_j$.
    For $\Gamma_{\mathrm{blanket}}$, transposing each scalar summand and swapping the summation order yields
    $\Gamma_{\mathrm{blanket}}
    =
    \sum_{j \in \partial\mathcal{B}}
    e_j^\top \psi_j$
    with
    $\psi_j(x)
    \coloneqq
    \sum_{i \in \mathcal{B}} A_{ij}(x)^\top \nabla_{x_i}\varphi(x) \in \mathbb{R}^{n_j}$.
    Here the sum ranges over all of~$\partial\mathcal{B}$.
    Assumption~\ref{ass:port-insulated} gives $\partial\mathcal{B} \subseteq \mathcal{P}$, reducing this to
    $\sum_{j \in \partial\mathcal{B} \cap \mathcal{P}} e_j^\top \psi_j$ with $\psi_j \coloneqq \psi_j$.       
    
    Since $\partial\mathcal{S}_\gamma$ is compact by
    Assumption~\ref{ass:compact}, each continuous coefficient function is
    bounded there.
    Define
    \[
      \bar\psi_j
      \;\coloneqq\;
      \sup_{x \in \partial\mathcal{S}_\gamma}
      \|\psi_j(x)\|,
      \qquad
      \bar\phi_j
      \;\coloneqq\;
      \sup_{x \in \partial\mathcal{S}_\gamma}
      \|\phi_j(x)\|,
    \]
    and let
    $c_1
    \coloneqq
    \bigl(
    \sum_{j \in \partial\mathcal{B} \cap \mathcal{P}}
    (\bar\psi_j + \bar\phi_j)^2
    \bigr)^{1/2}$.
    Then, by Cauchy--Schwarz on each inner product and
    subsequently on the discrete sum,
    \begin{equation*}
      \begin{aligned}
        |\Gamma_{\mathrm{blanket}}(x)|
        +
        |\Gamma_{\mathrm{coupling}}(x)|
        &\le
        \sum_{j \in \partial\mathcal{B} \cap \mathcal{P}}
        (\bar\psi_j + \bar\phi_j)\,\|e_j(x)\|
        \\
        &\le
        c_1\,\|e_{\partial\mathcal{B}}(x)\|.
      \end{aligned}
    \end{equation*}
    Likewise,
    $L_f \varphi|_{\text{intra-}\mathcal B}$ is continuous on the compact set
    $\partial\mathcal S_\gamma$, so:
    \begin{equation}\label{eq:intra-barrier-sup}
        M
    \coloneqq
    \sup_{x\in\partial\mathcal S_\gamma}
    \bigl|L_f \varphi|_{\text{intra-}\mathcal B}(x)\bigr|
    <\infty.
    \end{equation}
    Substituting these two bounds into the previous estimate yields
    \[
    (-L_f h_\gamma(x))^+
    \le
    M + c_1\,\|e_{\partial\mathcal B}(x)\|,
    \qquad
    x\in \partial\mathcal S_\gamma,
    \]
    which is \eqref{eq:adverse-drift-envelope} as desired.
\end{proof}

\subsection{Proof of Corollary~\ref{cor:strong-feasibility} (Feasibility under strong structural condition)}%
\label{sec:proof-strong-feasibility}

\begin{proof}
    Since Definition~\ref{def:strong-structural} holds, the intra-barrier
    residual in~\eqref{eq:residual-drift} vanishes identically on all of~$\mathcal{D}$, not just
    on~$\mathcal{Z}_{\partial\mathcal{B}}$.
    The constant $M = 0$ by definition in \eqref{eq:intra-barrier-sup}.
    Therefore, 
    Lemma~\ref{lem:adverse-drift-envelope} gives 
    $$(-L_f h_\gamma(x))^+ \leq c_1\|e_{\partial \mathcal B}(x)\|$$
    under Assumption~\ref{ass:local-dep}, \ref{ass:port-insulated}, and \ref{ass:compact}.
    By setting $c_0 \coloneqq c_1$, Assumption~\ref{ass:dom-degen} is satisfied, and the hypotheses of Theorem~\ref{thm:pt-feasibility} hold with $U = \partial \mathcal S_\gamma$.
    By Assumption~\ref{ass:blanket-coercivity}:
    \begin{equation*}
      \frac{(-L_f h_\gamma)^+}{\|L_G h_\gamma\|}
      \;\leq\;
      \frac{c_1\,\|e_{\partial\mathcal{B}}\|}
           {(\sigma_{\mathcal{B}}/\gamma)\,
            \|e_{\partial\mathcal{B}}\|}
      =
      \frac{\gamma\, c_1}{\sigma_{\mathcal{B}}}
    \end{equation*}
    uniformly on
    $\partial\mathcal{S}_\gamma$.
\end{proof}

\subsection{Verifying Blanket Coercivity}
\label{sec:rem-verify-coercivity}
\begin{remark}[Verifying blanket coercivity]%
\label{rem:verify-coercivity}
When $|\partial\mathcal{B} \cap \mathcal{P}| = 1$, say $\{p\}$, condition~\eqref{eq:blanket-coercivity} holds with
\[
\sigma_{\mathcal B}=\beta_p\,\underline{\sigma}_p,
\qquad
\underline{\sigma}_p \coloneqq \inf_{x\in\mathcal D}\sigma_{\min}(G_p(x)),
\]
provided $\underline{\sigma}_p>0$, i.e., $G_p(x)$ has full row rank uniformly on~$\mathcal D$.
For multiple blanket ports, one additionally needs the weighted terms $\beta_j\, e_j^\top G_j$ not to cancel across ports.
A sufficient condition is:
\[
G_j(x)G_j(x)^\top \succeq \underline{\sigma}_j^2 I_{n_j}
\quad
\text{for all }x\in\mathcal D,\; j\in \partial\mathcal{B}\cap\mathcal P,
\]
together with pairwise orthogonality
$\operatorname{im}(G_j(x)) \perp \operatorname{im}(G_k(x))$
with $\forall\,j\neq k$,
which guarantees that the summands occupy independent directions in~$\mathbb R^m$ and cannot cancel.
\end{remark}

To see why blockwise rank alone does not suffice, consider the
LC ladder of Example~\ref{ex:lc-ladder} with a single shared
input channel $m = 1$ and $G = (0,\, 0,\, 1,\, 1)^\top$.
Each block satisfies $G_j G_j^\top = 1 \geq \sigma_j^2$, yet
$L_G h_\gamma
= -\gamma^{-1}(\beta_1\, i_{L_1} + \beta_2\, i_{L_2})$
vanishes on the surface
$\{\beta_1\, i_{L_1} = -\beta_2\, i_{L_2}\}$.
The set where $L_G h_\gamma$ vanishes is strictly
larger than the set
$\{i_{L_1} = i_{L_2} = 0\}$ predicted by the graph topology.
In contrast, the dual-input map
$G = \bigl(\begin{smallmatrix}
0 & 0 & 1& 0 \\ 0 & 0 & 0 & 1
\end{smallmatrix}\bigr)^\top$
assigns each port its own channel, giving
$\operatorname{im}(G_{L_1}) \perp \operatorname{im}(G_{L_2})$
and restoring blanket coercivity with
$\sigma_{\mathcal{B}}
= \min(\beta_1\, \sigma_1,\; \beta_2\, \sigma_2)$ with $\sigma_j \coloneqq \sigma_{\min}(G_j(x))$.

\end{document}